\documentclass[
3p,
]{elsarticle}

\usepackage{verbatim}
\usepackage{float}
\usepackage{pgfplots}
	\pgfplotsset{compat=1.12}
	
\usepackage{tikz}
 	\usetikzlibrary{decorations.markings}
 	\usetikzlibrary{calc}
\usepackage{amsmath,amsthm,amssymb}
\usepackage[bookmarksdepth=3]{hyperref}
	\hypersetup{pdfauthor=author}
\usepackage[capitalize]{cleveref}
	\crefname{equation}{}{}
	\crefformat{section}{\textsection#2#1#3}
	\crefformat{appendix}{\textsection#2#1#3}

\theoremstyle{plain}

	\newtheorem{theorem}{Theorem}[section]
	\newtheorem{lemma}[theorem]{Lemma}
	\newtheorem{proposition}[theorem]{Proposition}
	\newtheorem{corollary}[theorem]{Corollary}

\theoremstyle{definition}

	\newtheorem{remark}[theorem]{Remark}
	\newtheorem{example}[theorem]{Example}

\newcommand{\N}{\mathbb{N}}
\newcommand{\Z}{\mathbb{Z}}
\newcommand{\R}{\mathbb{R}}
\newcommand{\C}{\mathbb{C}}
\newcommand{\T}{\mathbb{T}}

\newcommand{\cB}{\mathcal{B}}

\newcommand{\cH}{\mathcal{H}}

\newcommand{\sign}{\mathrm{sign}\,}

\renewcommand{\Re}{\mathrm{Re}\,}
\renewcommand{\Im}{\mathrm{Im}\,}
\newcommand{\Tr}{\mathrm{Trace}\,}
\newcommand{\ind}{\mathrm{ind}\,}
\newcommand{\dis}{\sigma_{\mathrm{dis}}}
\newcommand{\ess}{\sigma_{\mathrm{ess}}}
\newcommand{\Arg}{\mathrm{Arg}\,}

\newcommand{\Usuz}{U_{\textnormal{suz}}}

\newcommand{\textbi}[1]{\textit{\textbf{#1}}}

\newcommand\danger[1]{%
 \makebox[1em][c]{%
 \makebox[0pt][c]{\raisebox{.3em}{\textcolor{red}{\tiny \textbf{#1}}}}%
 \makebox[0pt][c]{\textcolor{red}{\fontsize{35}{45}\selectfont $\bigtriangleup$}}}}%

\newcounter{MyCounter}
\renewcommand\theMyCounter{\arabic{MyCounter}}
\newcommand{\dangernote}{\refstepcounter{MyCounter}\textsuperscript{\theMyCounter} \marginpar{\,\,\,\,\,\,{\footnotesize \danger{\theMyCounter}}}}

\usepackage{totcount}
\regtotcounter{MyCounter}
\usepackage{colortbl}

\begin{document}

\begin{frontmatter}

\title
{
The bulk-edge correspondence for the split-step quantum walk on the one-dimensional integer lattice
}

\author[Shinshu1]{Yasumichi Matsuzawa}
	\ead{myasu@shinshu-u.ac.jp}
\author[Hokkaido]{Motoki Seki}
	\ead{seki@math.sci.hokudai.ac.jp}
\author[Shinshu2]{Yohei Tanaka\corref{corresponding}}
	\ead{20hs602a@shinshu-u.ac.jp}

\cortext[corresponding]{Corresponding author}

\address[Shinshu1]{Department of Mathematics, Faculty of Education, Shinshu University, 6-Ro, Nishi-nagano,
Nagano 380-8544, Japan}
\address[Shinshu2]{Division of Mathematics and Physics, Faculty of Engineering, Shinshu University, Wakasato, Nagano 380-8553, Japan}
\address[Hokkaido]{Department of Mathematics, Faculty of Science, Hokkaido University, Kita 10, Nishi 8, Kita-Ku, Sapporo, Hokkaido, 060-0810, Japan}

\begin{abstract}
Suzuki's split-step quantum walk on the one-dimensional integer lattice can be naturally viewed as a chirally symmetric quantum walk. Given the unitary time-evolution of such a chirally symmetric quantum walk, we can separately introduce well-defined indices for the eigenvalues $\pm 1.$ The bulk-edge correspondence for Suzuki's split-step quantum walk is twofold. Firstly, we show that the multiplicities of the eigenvalues $\pm 1$ coincide with the absolute values of the associated indices. Note that this can be viewed as the symmetry protection of bound states, and that the indices we consider are robust in the sense that these depend only on the asymptotic behaviour of the parameters of the given model. Secondly, we show that that such bound states exhibit exponential decay at spatial infinity. 
\end{abstract}

\end{frontmatter}

%\tableofcontents

% ------------------------------------------------------------------------------------------------------------ %
% New Section                                                                                                  %
% ------------------------------------------------------------------------------------------------------------ %
\section{Introduction}

Quantum walk theory is widely recognised as a natural quantum-mechanical counterpart of the classical random walk theory \cite{Gudder-1988,Aharonov-Davidovich-Zagury-1993,Meyer-1996,Ambainis-Bach-Nayak-Vishwanath-Watrous-2001}. In this paper, we shall focus on the well-known discrete-time quantum, known as the \textbi{split-step quantum walk} on the one-dimensional integer lattice $\Z.$ This is a two-state quantum walk model, and so  $\ell^2(\Z, \C^2) = \ell^2(\Z) \oplus \ell^2(\Z)$ is the underlying state Hilbert space. The primitive form of the split-step quantum walk discussed in \cite{Kitagawa-Rudner-Berg-Demler-2010, Kitagawa-Broome-Fedrizzi-Rudner-Berg-Kassal-Aspuru-Demler-White-2012, Kitagawa-2012} can be characterised by the following unitary time-evolution;
\begin{equation}
\label{equation: kitagawa split-step quantum walk}
U_{\textnormal{kit}} := 
\begin{pmatrix}
1 & 0 \\
0 & L
\end{pmatrix}
\begin{pmatrix}
\cos \theta_2 & - \sin \theta_2 \\
\sin \theta_2 & \cos \theta_2
\end{pmatrix}
\begin{pmatrix}
L^* & 0 \\
0 & 1
\end{pmatrix}
\begin{pmatrix}
\cos \theta_1 & - \sin \theta_1 \\
\sin \theta_1 & \cos \theta_1
\end{pmatrix},
\end{equation}
where $L$ is the bilateral left-shift operator on $\ell^2(\Z),$ and where $\theta_j = (\theta_j(x))_{x \in \Z}$ is a $\R$-valued sequence for each $j = 1,2.$ We refer to \cref{equation: kitagawa split-step quantum walk} as the time-evolution operator of \textbi{Kitagawa's split-step quantum walk} throughout this paper, where each trigonometric sequence is viewed as a bounded multiplication operator on $\ell^2(\Z).$ It turns out that the evolution operator of Kitagawa's split-step quantum walk can be naturally generalised to that of \textbi{Suzuki's split-step quantum walk} \cite{Fuda-Funakawa-Suzuki-2017,Fuda-Funakawa-Suzuki-2018,Fuda-Funakawa-Suzuki-2019,Tanaka-2020,Narimatsu-Ohno-Wada-2021}, given explicitly by the following formula;
\begin{equation}
\label{equation: suzuki split-step quantum walk}
U_{\textnormal{suz}} := 
\begin{pmatrix}
p & q L \\
L^*q^*  & -p(\cdot - 1)
\end{pmatrix}
\begin{pmatrix}
a & b^* \\
b & -a
\end{pmatrix}, \qquad 
\varGamma := 
\begin{pmatrix}
p & q L \\
L^*q^*  & -p(\cdot - 1)
\end{pmatrix},
\end{equation}
where we always assume that $\R$-valued sequences $p = (p(x))_{x \in \Z}, a = (a(x))_{x \in \Z}$ and $\C$-valued sequences $q = (q(x))_{x \in \Z}, b = (b(x))_{x \in \Z}$ satisfy $p(x)^2 + |q(x)|^2 = 1$ and $a(x)^2 + |b(x)|^2 = 1$ for each $x \in \Z.$ Indeed, it is not difficult to show that \cref{equation: kitagawa split-step quantum walk} can be made unitarily equivalent to \cref{equation: suzuki split-step quantum walk}, provided that we appropriately define $p,q,a,b$ in terms of $\theta_1, \theta_2$ (see \cref{lemma: equivalence of Ukit} for details).

A notable advantage of using Suzuki's split-step quantum walk lies in the fact that the associated time-evolution operator $U = U_{\textnormal{suz}}$ exhibits chiral symmetry with respect to the unitary self-adjoint operator $\varGamma$ defined by the second equality in \cref{equation: suzuki split-step quantum walk};
\begin{equation}
\label{equation: chiral symmetry}
U^* = \varGamma U \varGamma.
\end{equation}
In fact, given any pair $(\varGamma, U)$ of an abstract unitary operator $U$ and an abstract unitary self-adjoint operator $\varGamma$ satisfying \cref{equation: chiral symmetry}, we can introduce the following two indices according to \cite{Cedzich-Geib-Grunbaum-Stahl-Velazquez-Werner-Werner-2018,Cedzich-Geib-Stahl-Velazquez-Werner-Werner-2018,Cedzich-Geib-Werner-Werner-2021}:
\begin{equation}
\label{equation: definition of symmetry indices}
\ind_+(\varGamma, U) := \Tr(\varGamma|_{\ker(U - 1)}), \qquad 
\ind_-(\varGamma, U) := \Tr(\varGamma|_{\ker(U + 1)}),
\end{equation}
where $\ker(U -1)$ and $\ker(U + 1)$ are $\varGamma$-invariant subspaces by \cref{equation: chiral symmetry}. Note that $\ind_+(\varGamma, U)$ is well-defined, if the essential spectrum of $U,$ defined by $\ess(U) := \{z \in \C \mid U - z \mbox{ is not Fredholm}\},$ does not contain $1.$ Indeed, the restriction $\varGamma|_{\ker(U - 1)}$ becomes a unitary self-adjoint operator on the finite-dimensional vector space $\ker(U - 1)$ in this case. Similarly, $\ind_-(\varGamma, U)$ is well-defined, if $-1 \notin \ess(U).$ The following \textit{bulk-edge correspondence} is one of the main results of the present article;

\begin{theorem}
\label{theorem: baby bulk-edge correspondence}
Let $U = \Usuz$ be the evolution operator of Suzuki's split-step quantum walk given by \cref{equation: suzuki split-step quantum walk}, and let us assume the existence of the following two-sided limits:
\begin{align}
\label{equation: anisotropic assumption}
&p(\pm \infty) := \lim_{x \to \pm \infty} p(x)  \in (-1,1), & &a(\pm \infty) := \lim_{x \to \pm \infty} a(x)  \in (-1,1).
\end{align}
Let $|p(x)| < 1$ and $|a(x)| < 1$ for each $x \in \Z.$ Then $p(\star) \neq  \pm a(\star)$ for each $\star  = \pm \infty$ if and only if $\pm 1 \notin \ess(U).$ In this case, the following two assertions hold true:
\begin{enumerate}[(i)]
\item We have $\dim \ker(U \mp 1) = |\ind_\pm(\varGamma,U)|,$ where
\begin{equation}
\label{equation: pm indices for anisotropic ssqw}
\ind_\pm(\varGamma, U) = 
\begin{cases}
+1, & p(-\infty) \mp a(-\infty) < 0 < p(+\infty) \mp a(+\infty), \\
-1, & p(+\infty) \mp a(+\infty) < 0 < p(-\infty) \mp a(-\infty), \\
0, & \mbox{otherwise}.
\end{cases}
\end{equation}

\item If $(-1)^j(p(-\infty) \mp a(-\infty)) < 0 < (-1)^j(p(+\infty) \mp a(+\infty))$ for some $j=1,2,$ then any non-zero vector $\Psi \in \ker(U \mp 1)$ admits the following unique representation;
\begin{equation}
\label{equation: eigenstate for the anosotropic case}
\Psi =
\begin{pmatrix}
\frac{a \mp (-1)^j}{b} \psi \\
\psi
\end{pmatrix}, \qquad \psi \in \ker\left(L + \frac{p + (-1)^j}{q}\frac{a \mp (-1)^j}{b}\right).
\end{equation}
Moreover, the eigenstate $\Psi$ characterised by \cref{equation: eigenstate for the anosotropic case} exhibits exponential decay. More precisely, there exist positive constants $c^\downarrow_\pm, c^\uparrow_\pm, \kappa^\downarrow_\pm, \kappa^\uparrow_\pm, x_\pm,$ such that
\begin{equation}
\label{equation: exponential decay in the anisotropic case}
\kappa^\downarrow_\pm e^{- c^\downarrow_\pm |x|} \leq \|\Psi(x)\|^2 \leq \kappa^\uparrow_\pm e^{- c^\uparrow_\pm |x|}, \qquad |x| \geq x_\pm.
\end{equation}
\end{enumerate}
\end{theorem}

Note first the equality $\dim \ker(U \mp 1) = |\ind_\pm(\varGamma, U)|$ in \cref{theorem: baby bulk-edge correspondence}(i) can be understood as the protection of eigenstates corresponding to $\pm 1$ by chiral symmetry. Here, the robustness of $\ind_\pm(\varGamma, U)$ is ensured by the formula \cref{equation: pm indices for anisotropic ssqw} which depends only on the asymptotic values \cref{equation: anisotropic assumption}. It is also shown in \cref{theorem: baby bulk-edge correspondence}(ii) that such symmetry protected eigenstates can be uniquely characterised by the explicit formula \cref{equation: eigenstate for the anosotropic case}, and that they exhibit exponential decay in the sense of \cref{equation: exponential decay in the anisotropic case}.

\cref{theorem: baby bulk-edge correspondence} can be classified as an index theorem for $2$-phase chirally symmetric quantum walks on the one-dimensional integer lattice $\Z,$ since we assume the existence of the two-sided limits as in \cref{equation: anisotropic assumption}. Index theory of such $2$-phase quantum walks can be found in the extensive literature \cite{Cedzich-Grunbaum-Stahl-Velazquez-Werner-Werner-2016,Cedzich-Geib-Grunbaum-Stahl-Velazquez-Werner-Werner-2018,Cedzich-Geib-Stahl-Velazquez-Werner-Werner-2018,Suzuki-2019,Suzuki-Tanaka-2019, Matsuzawa-2020,Asahara-Funakawa-Seki-Tanaka-2020, Tanaka-2020, Cedzich-Geib-Werner-Werner-2021}. As such, \cref{theorem: baby bulk-edge correspondence} may not seem novel at first glance.

Note, however, that the ultimate purpose of the present article is to generalise both \cref{theorem: baby bulk-edge correspondence}(i),(ii) by replacing the $2$-phase assumption \cref{equation: anisotropic assumption} with the following significantly weakened assumption;
\begin{align}
\label{equation: bounded assumption}
&\sup_{x \in \Z} |p(x)| < 1, \qquad \sup_{x \in \Z} |a(x)| < 1.
\end{align}
For example, the new assumption \cref{equation: bounded assumption} will allow us to consider the case where $p$ and $a$ are periodic. Unlike the $2$-phase case, it seems unrealistic to extract any useful information about $\ess(\Usuz)$ from the new abstract assumption \cref{equation: bounded assumption}. As such, it is desirable to also generalise \cref{equation: definition of symmetry indices}. 

The present article is organised as follows. In \cref{section: preliminaries} we develop new index theory for abstract unitary operators $U$ exhibiting chiral symmetry in the sense of \cref{equation: chiral symmetry} in full generality. Note that this somewhat elementary construction is beyond the scope of the existing literature \cite{Cedzich-Geib-Grunbaum-Stahl-Velazquez-Werner-Werner-2018,Cedzich-Geib-Stahl-Velazquez-Werner-Werner-2018,Suzuki-2019,Tanaka-2020,Cedzich-Geib-Werner-Werner-2021}, since it makes use of neither the notion of a Fredholm operator, nor any local structure of the underlying Hilbert space. We show that our indices coincide with \cref{equation: definition of symmetry indices}, if the essential spectrum of $U$ has a spectral gap at $\pm 1.$ The purpose of \cref{section: bulk-edge correspondence} is to prove \cref{theorem: baby bulk-edge correspondence} with \cref{equation: anisotropic assumption} replaced by \cref{equation: bounded assumption} (see \cref{theorem: bulk-boundary correspondence} for more details). As we shall see, \cref{theorem: baby bulk-edge correspondence}(ii) is a natural extension of \cite[Theorem 5.1]{Fuda-Funakawa-Suzuki-2018} which states that non-trivial vectors in the so-called birth eigenspaces exhibit exponential decay. However, we do not make use of the spectral mapping theorem for chirally symmetric unitary operators discussed in \cite{Segawa-Suzuki-2016,Segawa-Suzuki-2019} unlike \cite{Fuda-Funakawa-Suzuki-2018}. The present article concludes with the summary and discussion in \cref{section: concluding remarks}. For example, it is shown in this section that the decay rates $c^\downarrow_\pm, c^\uparrow_\pm$ in \cref{equation: exponential decay in the anisotropic case} depend on the gaps of the essential spectrum under the $2$-phase assumption \cref{equation: anisotropic assumption}, and that the index formulas in \cite{Suzuki-Tanaka-2019, Matsuzawa-2020,Asahara-Funakawa-Seki-Tanaka-2020, Tanaka-2020} can be easily derived from \cref{theorem: baby bulk-edge correspondence}(i).

% ------------------------------------------------------------------------------------------------------------ %
% New Section                                                                                                  %
% ------------------------------------------------------------------------------------------------------------ %
\section{Indices for chirally symmetric unitary operators}
\label{section: preliminaries}

By operators we always mean everywhere-defined bounded linear operators between Banach spaces throughout this paper. Recall that the \textbi{(Fredholm) essential spectrum} of an operator $X$ on a Hilbert space $\cH$ is defined by $\ess(X) := \{z \in \C \mid X - z \mbox{ is not Fredholm}\}.$ If $X$ is normal, then $\ess(X) = \sigma(X) \setminus \dis(X),$ where $\dis(X)$ is the discrete spectrum of $X.$ Note that the equality $\ker X = \ker X^*X$ shall be repeatedly used without any further comment.

A \textbi{chiral pair} on $\cH$ is any pair $(\varGamma,U)$ of a unitary self-adjoint operator $\varGamma : \cH \to \cH$ and an operator $U: \cH \to \cH,$  satisfying the chiral symmetry condition \cref{equation: chiral symmetry}. Note that the underlying Hilbert space $\cH$ admits a $\Z_2$-grading of the form $\cH = \ker(\varGamma - 1) \oplus \ker(\varGamma + 1),$ and that $\varGamma = 1 \oplus (-1)$ with respect to this orthogonal decomposition, where $1$ denotes the identity operator on a Hilbert space throughout this paper. The operator $U$ can then be written as $U = R + iQ,$ where $R, Q$ are the real and imaginary parts of $U$ respectively. We have: 
\begin{align}
% ------------------------------ %
\label{equation: representation of R and Q}
R =  
\begin{pmatrix}
R_1 & 0 \\
0 & R_2
\end{pmatrix}_{\ker(\varGamma - 1) \oplus \ker(\varGamma + 1)}, \qquad 
Q = \begin{pmatrix}
0 & Q_2 \\
Q_1 & 0
\end{pmatrix}_{\ker(\varGamma - 1) \oplus \ker(\varGamma + 1)}.
\end{align}
Here, the first equality follows from the commutation relation $[\varGamma, R] := \varGamma R - R \varGamma = 0,$ whereas the second equality follows from the anti-commutation relation $\{\varGamma, Q\} := \varGamma Q + Q \varGamma = 0$ (see \cite[Lemma 2.2]{Suzuki-2019} for details). Since $R, Q$ are self-adjoint, we have $R_j^* = R_j$ for each $j=1,2,$ and $Q_2 = Q_1^*.$ The following formula shall be referred to as the \textbi{standard representation} of $U$ with respect to $\varGamma$ throughout this paper;
\begin{equation}
\label{equation: standard representation of U} 
U = 
\begin{pmatrix}
R_1 & iQ_2 \\
iQ_1 & R_2
\end{pmatrix}_{\ker(\varGamma - 1) \oplus \ker(\varGamma + 1)}.
\end{equation}
With \cref{equation: standard representation of U} in mind, we introduce the following formal indices:
\begin{align}
\label{equation: definition of pm indices}
\ind_\pm(\varGamma, U) &:= \dim \ker (R_1 \mp 1) - \dim \ker (R_2 \mp 1), \\
\ind(\varGamma, U) &:= \dim \ker Q_1 - \dim \ker Q_2.
\end{align}

\begin{lemma}
Given a chiral pair $(\varGamma, U)$ with \cref{equation: standard representation of U} being the standard representation of $U,$ we have 
\begin{align}
\label{equation: ker U and Rj}
\ker(U \mp 1) &= \ker(R_1 \mp 1) \oplus \ker(R_2 \mp 1), \\
\label{Equation: Rj and Qj}
\ker Q_j &= \ker(R_j - 1) \oplus \ker(R_j + 1), \qquad j=1,2.
\end{align}
Moreover, the following assertions hold true:
\begin{enumerate}[(i)]
\item The index $\ind_\pm(\varGamma,U)$ is a well-defined integer, if $\dim \ker(U \mp 1) < \infty.$ In this case, 
\begin{equation}
\label{equation: topological protection of bounded states}
|\ind_\pm(\varGamma,U)| \leq \dim \ker(U \mp 1).
\end{equation}
\item The index $\ind(\varGamma,U)$ is a well-defined integer, if $\dim \ker(U - 1) \oplus \ker(U + 1) < \infty.$ In this case, 
\begin{equation}
\label{Equation: Witten Index is Sum of pm Indies}
\ind(\varGamma,U) = \ind_+(\varGamma,U) + \ind_-(\varGamma,U).
\end{equation}
\end{enumerate}
\end{lemma}
Note that $\dim \ker(U \mp 1) < \infty$ is a weaker assumption than $\pm 1 \notin \ess(U).$ 
\begin{proof}
Since $U = R + iQ$ is unitary and since $[R,Q] = 0,$ we have $R^2 + Q^2 = 1.$  Firstly, this matrix equality implies $R_j^2 + Q_j^*Q_j = 1$ for each $j = 1,2,$ and so \cref{Equation: Rj and Qj} follows. Secondly, the same equality implies $(U \mp 1)^*(U \mp 1) = 2(1 \mp R).$ We obtain \cref{equation: ker U and Rj} from
\begin{equation}
\label{equation: ker U and R}
\ker(U \mp 1) = \ker(U \mp 1)^*(U \mp 1) =  \ker(1 \mp R) = \ker(R \mp 1),
\end{equation}
where $\ker(R \mp 1) = \ker(R_1 \mp 1) \oplus \ker(R_2 \mp 1),$ since $R = R_1 \oplus R_2.$

(i) It follows from  \cref{equation: ker U and Rj} that if $\dim \ker(U \mp 1) < \infty,$ then $\dim \ker(R_j \mp 1)< \infty$ for each $j=1,2,$ and so $\ind_\pm(\varGamma,U)$ is well-defined. We have
\[
|\ind_\pm(\varGamma,U)|
\leq  \dim\ker(R_1 \mp 1) + \dim \ker(R_2 \mp 1)
= \dim \ker(U \mp 1).
\]

(ii) It follows from \cref{Equation: Rj and Qj} that
\begin{equation}
\label{equation: ker Qj and Rj}
\dim \ker Q_j = \dim \ker(R_j - 1) + \dim \ker(R_j + 1), \qquad j = 1,2.
\end{equation}
If $\dim \ker(U - 1) \oplus  \ker(U + 1) < \infty,$ then $\dim \ker(R_j - 1) \oplus \ker(R_j + 1) < \infty$ for each $j=1,2$ by \cref{equation: ker U and R}. It follows from \cref{equation: ker Qj and Rj} that
\begin{align*}
\ind(\varGamma, U)
&=  \dim \ker Q_1 - \dim \ker Q_2 \\
&=  \dim \ker(R_1 - 1) + \dim \ker(R_1 + 1) - ( \dim \ker(R_2 - 1) + \dim \ker(R_2 + 1)) \\
&=  \dim \ker(R_1 - 1)  - \dim \ker(R_2 - 1) + \dim \ker(R_1 + 1) - \dim \ker(R_2 + 1) \\
&= \ind_+(\varGamma, U) + \ind_-(\varGamma, U).
\end{align*}
\end{proof}

\begin{lemma}
\label{lemma: unitary invariance of the indices}
Let $(\varGamma_0,U_0), (\varGamma,U)$ be two chiral pairs on Hilbert spaces $\cH_0,\cH$ respectively. If  $(\varGamma_0, U_0), (\varGamma, U)$ are \textbi{unitarily equivalent} in the sense that $(\varGamma_0, U_0) = (\epsilon^* \varGamma \epsilon, \epsilon^* U \epsilon)$ for some unitary operator $\epsilon : \cH_0 \to \cH,$ then the following assertions hold true:
\begin{enumerate}[(i)]
\item If $\dim \ker(U_0 \mp 1) = \dim \ker(U \mp 1)$ is finite, then $\ind_\pm(\varGamma_0, U_0) = \ind_\pm(\varGamma, U).$
\item If $\dim \ker(U_0 - 1) \oplus \ker(U_0 + 1) = \dim \ker(U - 1) \oplus \ker(U + 1)$ is finite, then $\ind(\varGamma_0, U_0) = \ind(\varGamma, U).$
\end{enumerate}
\end{lemma}
\begin{proof} 
The details of what follows can be found in the proof of \cite[Lemma 2]{Asahara-Funakawa-Seki-Tanaka-2020}. Firstly, there exists a unitary operator $\epsilon_j : \ker(\varGamma_0 + (-1)^j) \to \ker(\varGamma + (-1)^j)$ for each $j=1,2,$ such that the given unitary operator $\epsilon$ can be identified with the direct sum $\epsilon_{1} \oplus \epsilon_{2} : \ker(\varGamma_0 - 1) \oplus \ker(\varGamma_0 + 1) \to \ker(\varGamma - 1) \oplus \ker(\varGamma + 1).$ Secondly, if $U$ admits the standard representation of the form \cref{equation: standard representation of U}, then the standard representation of $U_0$ is given by the following formula;
\[
U_0 = 
\begin{pmatrix}
\epsilon^*_{1}  R_1 \epsilon_{1} & i \epsilon^*_{1}  Q_2 \epsilon_{2} \\
i \epsilon^*_{2}  Q_1 \epsilon_{1} & \epsilon^*_{2} R_2 \epsilon_{2}\\
\end{pmatrix}_{\ker(\varGamma_0 - 1) \oplus \ker(\varGamma_0 + 1)}.
\]

(i) The claim follows from $\dim \ker(R_j \mp 1) = \dim \ker(\epsilon^*_{j}(R_j \mp 1)\epsilon_{j}) = \ker(\epsilon^*_{j}R_j \epsilon_{j} \mp 1)$ for each $j=1,2.$

(ii) We get
\begin{align*}
\dim \ker Q_1 &= \dim \ker Q_1^* Q_1 = \dim \ker(\epsilon^*_{1}Q_1^* Q_1\epsilon_{1}) = \dim \ker(\epsilon_2^*Q_1\epsilon_{1})^* ( \epsilon_2^*Q_1\epsilon_{1}) = \dim \ker( \epsilon_2^*Q_1\epsilon_{1}), \\
\dim \ker Q_2 &= \dim \ker Q_2^* Q_2 = \dim \ker(\epsilon^*_{2}Q_2^* Q_2\epsilon_{2}) = \dim \ker(\epsilon_1^*Q_2\epsilon_{2})^* ( \epsilon_1^*Q_2\epsilon_{2}) = \dim \ker( \epsilon_1^*Q_2\epsilon_{2}).
\end{align*}
\end{proof}

\begin{proposition}
\label{theorem: suzuki formula}
Let $(\varGamma, U)$ be a chiral pair, and let $\varGamma' := \varGamma U.$ Then the following assertions hold true:
\begin{enumerate}[(i)]
\item If $U$ admits the standard representation of the form \cref{equation: standard representation of U}, then
\begin{equation}
\label{equation: representation of Rj mp 1}
\ker(R_j \mp 1) 
=  \ker(U \mp 1) \cap \ker(\varGamma + (-1)^j)
= \ker(\varGamma + (-1)^j) \cap \ker(\varGamma' \mp (-1)^{j+1}), \qquad j=1,2.
\end{equation}
\item The pair $(\varGamma', U)$ is also a chiral pair. Moreover, if $\ker(U \mp 1)$ is finite-dimensional, then 
\begin{equation}
\label{equation: index formulas with different chiral symmetry conditions}
\ind_\pm(\varGamma, U) = \pm \ind_\pm(\varGamma', U).
\end{equation}

\item If $\ker(U - 1) \oplus \ker(U + 1)$ is finite-dimensional, then 
\begin{equation}
%\label{equation: index formulas with different chiral symmetry conditions}
\ind(\varGamma', U) = \ind_+(\varGamma, U) - \ind_-(\varGamma, U).
\end{equation}
\end{enumerate}
\end{proposition}
\begin{proof}
(i)  We have
\[
U \mp 1  =
\begin{pmatrix}
R_1 \mp 1 & i Q_2 \\
i Q_1 & R_2 \mp 1\\
\end{pmatrix}.
\]
It follows from this equality that
\[
\ker(U \mp 1) \cap \ker(\varGamma + (-1)^j) = \ker(R_j \mp 1) \cap \ker Q_j = \ker(R_j \mp 1),
\]
where the last equality follows from \cref{Equation: Rj and Qj}. Similarly, we have
\[
\varGamma' \mp  1
=  \varGamma U \mp  1
=
\begin{pmatrix}
1 & 0 \\
0 & -1
\end{pmatrix}
\begin{pmatrix}
R_1  & i Q_2 \\
i Q_1 & R_2 \\
\end{pmatrix} \mp  1
=
\begin{pmatrix}
R_1 \mp 1  & i Q_2 \\
-i Q_1 & -(R_2  \pm 1)\\
\end{pmatrix}.
\]
We obtain
\begin{align*}
\ker(\varGamma -1) \cap \ker(\varGamma' \mp  1)  = \ker(R_1 \mp 1) \cap \ker Q_1 = \ker(R_1 \mp 1), \\
\ker(\varGamma +1) \cap \ker(\varGamma' \mp  1)  = \ker(R_2 \pm 1) \cap \ker Q_2 = \ker(R_2 \pm 1).
\end{align*}
The above identities can be written as the single formula \cref{equation: representation of Rj mp 1}.

(ii) Note that $(\varGamma', U)$ is a chiral pair, since $\varGamma' U \varGamma' = \varGamma' (\varGamma \varGamma') \varGamma' = \varGamma' \varGamma  = U^*.$ Let $\ker(U \mp 1)$ be finite-dimensional, and let
\begin{align*}
m_{j,\pm} &:= \dim \left(\ker(\varGamma + (-1)^j) \cap \ker(\varGamma' \mp (-1)^{j+1}) \right), \\
m'_{j,\pm} &:= \dim \left(\ker(\varGamma' + (-1)^j) \cap \ker(\varGamma \mp (-1)^{j+1}) \right).
\end{align*}
It follows from (i) that $\ind_\pm(\varGamma,U) = m_{1,\pm} - m_{2,\pm}$ and $\ind_\pm(\varGamma',U) = m'_{1,\pm} - m'_{2,\pm}.$ The formula \cref{equation: index formulas with different chiral symmetry conditions} is an immediate consequence of the following equalities:
\begin{align*}
m'_{1,+} &=  m_{1,+}, & m'_{2,+} &= m_{2,+}, \\
m'_{1,-} &=  m_{2,-}, & m'_{2,-} &= m_{1,-}.
\end{align*}

(iii) This follows from (i) and (ii).
\end{proof}

\begin{remark}
If $\dim \ker(U \mp 1) < \infty,$ then it follows from the first equality in \cref{equation: representation of Rj mp 1} that
\[
\ind_\pm(\varGamma, U) = \dim \ker(U \mp 1) \cap \ker(\varGamma - 1) -  \dim \ker(U \mp 1) \cap \ker(\varGamma + 1).
\]
Note that the right hand side can be written as $\Tr(\varGamma|_{\ker(U \mp 1)}),$ where $\ker(U - 1)$ and $\ker(U + 1)$ are $\varGamma$-invariant subspaces by the chiral symmetry condition \cref{equation: chiral symmetry}. That is, the previously mentioned formula \cref{equation: definition of symmetry indices} is consistent with \cref{equation: definition of pm indices}.
\end{remark}

\begin{corollary}
Let $(\varGamma, U)$ be a chiral pair, and let $\ker(U - 1) \oplus \ker(U + 1)$ be finite-dimensional. Then we have the following formulas:
\begin{align}
\label{Equation: Minus Abstract Coin} 
\ind_\pm(\varGamma, -U) &= \ind_\mp(\varGamma,U),           &  \ind(\varGamma,-U) &= \ind(\varGamma,U), \\
\ind_\pm(-\varGamma,U) &= - \ind_\pm(\varGamma,U),           &  \ind(-\varGamma,U) &= - \ind(\varGamma,U).
\end{align}
\end{corollary}
\begin{proof}
If $U$ admits the standard representation of the form \cref{equation: standard representation of U} with respect to $\varGamma,$ then the standard representation of $-U$ is
\[
-U = 
\begin{pmatrix}
-R_1 & -iQ_2 \\
-iQ_1 & -R_2
\end{pmatrix}_{\ker(\varGamma - 1) \oplus \ker(\varGamma + 1)}.
\]
It follows that $\ind_\pm(\varGamma, -U) = \ind_\mp(\varGamma,U),$ and so $\ind(\varGamma,-U) = \ind(\varGamma,U)$ by \cref{Equation: Witten Index is Sum of pm Indies}. Similarly, the standard representation of $U$ with respect to $-\varGamma$ is
\[
U = 
\begin{pmatrix}
R_2 & iQ_1 \\
iQ_2 & R_1
\end{pmatrix}_{\ker(\varGamma + 1) \oplus \ker(\varGamma - 1)}.
\]
We have $\ind_\pm(-\varGamma, U) = \dim \ker (R_2 \mp 1) - \dim \ker (R_1 \mp 1) = -\ind_\mp(\varGamma, U),$ and so $\ind(-\varGamma,U) = - \ind(\varGamma,U).$
\end{proof}

% ------------------------------------------------------------------------------------------------------------ %
% New Section                                                                                                  %
% ------------------------------------------------------------------------------------------------------------ %
\section{The Bulk-edge Correspondence}
\label{section: bulk-edge correspondence}

We are now in a position to state the following generalisation of \cref{theorem: baby bulk-edge correspondence};

\begin{theorem}
\label{theorem: bulk-boundary correspondence}
Let $U = \Usuz$ be the evolution operator of Suzuki's split-step quantum walk given by \cref{equation: suzuki split-step quantum walk}, and let \cref{equation: bounded assumption} hold true. Let us introduce the following notation:
\begin{align}
\label{equation: definition of F}
\Lambda(\kappa) &:= \frac{1+\kappa}{1-\kappa}, \qquad \kappa \in (-1,1), \\
% -------- %
\label{equation: definition of delta jpm}
\delta_{j,\pm}(y) &:= \frac{\sqrt{ \Lambda((-1)^j p(y))\Lambda(\mp (-1)^j a(y))}}{\pm  e^{i (\Arg q(y) + \Arg b(y))}},  \qquad y \in \Z, \\
% -------- %
\label{equation: definition of Delta jpm}
\Delta_{j, \pm} &:= \sum_{x =1}^\infty \left(\prod_{y=1}^{x} |\delta_{j,\pm}(-y)|^{-2}\right) + \sum_{x =1}^\infty \left(\prod_{y=0}^{x-1} |\delta_{j,\pm}(y)|^2\right),
\end{align}
where $j=1,2,$ and where $\Arg w$ denotes the principal argument of a non-zero complex number $w.$ Then $\Delta_{1, \pm}$ and $\Delta_{2, \pm}$ cannot be simultaneously finite, and the following assertions hold true:
\begin{enumerate}[(i)]
\item The dimension of $\ker(U \mp 1)$ is at most $1.$ More explicitly, we have 
\begin{align}
\label{equation: bulk-edge correspondence for ssqw}
|\ind_\pm(\varGamma,U)| &= \dim \ker(U \mp 1), \\
\label{equation: pm index for ssqw}
\ind_\pm(\varGamma,U) &= 
\begin{cases}
+1, & \Delta_{1, \pm} < \infty, \\
-1, & \Delta_{2, \pm} < \infty, \\
0,  & \Delta_{1, \pm} =  \Delta_{2, \pm} = \infty.
\end{cases}
\end{align}

\item If $\Delta_{j, \pm} < \infty$ for some $j=1,2,$ then we have the linear isomorphism $\tau_{j, \pm}: \ker(L - \delta_{j,\pm}) \to \ker(U \mp 1)$ defined by
\begin{equation}
\label{Equation2: Linear Isomorphism}
\tau_{j, \pm}(\psi):=
\begin{pmatrix}
\frac{\sqrt{\Lambda(\mp (-1)^j a )}}{\mp (-1)^j e^{i \Arg b}} \psi \\
\psi
\end{pmatrix}, \qquad \psi \in \ker(L - \delta_{j,\pm}).
\end{equation}
That is, $\dim \ker(L - \delta_{j,\pm}) = \dim \ker(U \mp 1) = 1$ according to \crefrange{equation: bulk-edge correspondence for ssqw}{equation: pm index for ssqw}.

\item For each $j=1,2,$ let
\begin{align}
\delta^{\downarrow}_{j, \pm} &:= \min\left\{\liminf_{x \to \infty} \left(\prod_{y=1}^{x} |\delta_{j, \pm}(-y)|^{-2}\right)^{1/x}, \qquad  \liminf_{x \to \infty} \left(\prod_{y=0}^{x-1} |\delta_{j, \pm}(y)|^2\right)^{1/x} \right\}, \\
\delta^{\uparrow}_{j, \pm} &:= \max\left\{\limsup_{x \to \infty} \left(\prod_{y=1}^{x} |\delta_{j, \pm}(-y)|^{-2}\right)^{1/x}, \qquad \limsup_{x \to \infty} \left(\prod_{y=0}^{x-1} |\delta_{j, \pm}(y)|^2\right)^{1/x} \right\}, \\
\Lambda_{j, \pm}^\downarrow & := \inf_{x \in \Z} \Lambda(\mp (-1)^j a(x)) + 1, \\
\Lambda_{j, \pm}^\uparrow & := \sup_{x \in \Z} \Lambda(\mp (-1)^j a(x)) + 1.
\end{align}
If $0 < \delta^{\downarrow}_{j,\pm} \leq \delta^{\uparrow}_{j,\pm} < 1$ for some $j=1,2,$  then $\Delta_{j, \pm} < \infty.$ Moreover, in this case, for any $\epsilon > 0$ satisfying $0 < \delta^{\downarrow}_{j,\pm} - \epsilon  < \delta^{\uparrow}_{j,\pm} + \epsilon < 1,$ there exists $x_\pm \in \N$ with the property that if $\psi \in \ker(L - \delta_{j,\pm})$ is a non-zero vector, then $\Psi := \tau_{j,\pm}(\psi)$ given by \cref{Equation2: Linear Isomorphism} exhibits the following exponential decay;
\begin{equation}
\label{equation2: exponential decay}
\Lambda_{j, \pm}^\downarrow \left(\delta^{\downarrow} _{j, \pm}- \epsilon \right)^{|x|} \leq \frac{\|\Psi(x)\|^2}{|\psi(0)|^2} \leq  \Lambda_{j, \pm}^\uparrow \left(\delta^{\uparrow}_{j, \pm} + \epsilon\right)^{|x|}, \qquad |x| \geq x_\pm.
\end{equation}
\end{enumerate}
\end{theorem}

\begin{remark}
\label{remark: bulk-boundary correspondence}
We have the following remarks:
\begin{enumerate}[(i)]
\item Note that the function $\Lambda$ defined by \cref{equation: definition of F} is a bijection from $(-1,1)$ onto $(0, \infty),$ and that the graph of $\Lambda$ is given by the following figure;
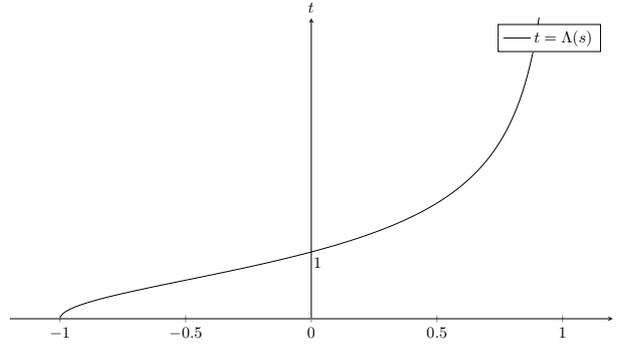
\begin{figure}[H]
\centering
\label{graph: graph of g}
\begin{tikzpicture}[scale=0.6]
\begin{axis}[axis lines=center,width = 0.9\textwidth,height = 0.5\textwidth,xlabel=$s$, xlabel style={anchor = west}, ylabel=$t$, ylabel style={anchor = south}, xtick= {-1, -1/2, 0, 1/2, 1}, extra x ticks={0}, xticklabel style={anchor = north}, xmin= -1.2, xmax=1.2, ymin= 0, ymax=4.5, ytick= {0,1}, yticklabel style = {anchor = north west}, ytick style={draw=none}]
\addplot[samples=200, domain=-0.999:0.999]{sqrt((1 + x)/(1 - x))};
\addlegendentry{$t = \Lambda(s)$}
%\addplot[dashed,samples=200,domain=-0.9:0.9]{sqrt((1 - x)/(1 + x))};
%\addlegendentry{$y = f(-\kappa)$}
\end{axis}
\end{tikzpicture}
\caption{This figure represents the graph of $t = \Lambda(s).$}
\label{figure: graph of Lambda}
\end{figure}
That is, $\Lambda_{j,\pm}$ defined by \cref{equation: definition of Delta jpm} is either a finite positive number or $+\infty.$ The formula \cref{equation: pm index for ssqw} is a complete classification of $\ind_\pm(\varGamma,U),$ since $\Lambda_{1,\pm}$ and $\Lambda_{2,\pm}$ cannot be finite at the same time according to \cref{theorem: bulk-boundary correspondence}.

\item It is in general difficult to compute $\delta^{\downarrow}_{j, \pm}, \delta^{\uparrow}_{j, \pm},$ but the following estimates may be useful\footnote 
{
The estimate \crefrange{equation1: estimate for delta}{equation2: estimate for delta} can be easily proved by the well-known fact (see, for example, \cite[Theorem 3.37]{Rudin-1976}) that given a sequence $(\alpha(x))_{x \in \N}$ of positive numbers, we have
\[
\liminf_{x \to \infty} \frac{\alpha(x+1)}{\alpha(x)} 
\leq 
\liminf_{x \to \infty} \alpha(x)^{1/x}
\leq 
\limsup_{x \to \infty} \alpha(x)^{1/x}
\leq 
\limsup_{x \to \infty} \frac{\alpha(x+1)}{\alpha(x)}.
\]
}:
\begin{align}
\label{equation1: estimate for delta}
\delta^{\downarrow}_{j, \pm} &\geq \min\left\{\liminf_{x \to \infty} |\delta_{j, \pm}(-x)|^{-2}, \liminf_{x \to \infty} |\delta_{j, \pm}(x)|^{2}\right\}, \\
\label{equation2: estimate for delta}
\delta^{\uparrow}_{j, \pm} &\leq \max\left\{\limsup_{x \to \infty} |\delta_{j, \pm}(-x)|^{-2}, \limsup_{x \to \infty} |\delta_{j, \pm}(x)|^{2} \right\}.
\end{align}
\end{enumerate}
\end{remark}

% ------------------------------------------------------------------------------------------------------------ %
\subsection{Proof of the bulk-edge correspondence}
The purpose of the current section is to prove \cref{theorem: bulk-boundary correspondence}. We can then obtain \cref{theorem: baby bulk-edge correspondence} as an immediate corollary. In what follows we shall make use of the following obvious properties of $\Lambda$ without any further comment. For each $s,s' \in (-1,1),$ we have
\begin{align}
\label{equation1: properties of F}
&\Lambda(-s) = \Lambda(s)^{-1}, \\
\label{equation2: properties of F}
&\Lambda(s) \Lambda(s') = \Lambda\left( \frac{s + s'}{1 + ss'}\right), \\
\label{equation3: properties of F}
&\Lambda(s) \Lambda(s') \lesseqgtr 1 \mbox{ if and only if } s + s' \lesseqgtr 0, 
\end{align}
where $1 + ss' > 0$ in \cref{equation2: properties of F}, and where the notation $\lesseqgtr$ in \cref{equation3: properties of F} simultaneously denotes the three binary relations $>, =, <.$

\begin{lemma}
\label{lemma: intersection gives difference equation}
If \cref{equation: bounded assumption} holds true, then we have the following well-defined linear isomorphisms:
\begin{equation}
\label{Equation: Linear Isomorphism}
\ker(L - \delta_{j,\pm}) \ni \psi \longmapsto 
\begin{pmatrix}
\frac{a \mp (-1)^j}{b}\psi \\
\psi
\end{pmatrix} \in 
\ker(\varGamma + (- 1)^j) \cap \ker(\varGamma' \mp (-1)^{j+1}), \qquad j=1,2,
\end{equation}
where the bounded sequence $\delta_{j, \pm}$ is defined by \cref{equation: definition of delta jpm}.
\end{lemma}
That is,
\[
\ker(\varGamma + (- 1)^j) \cap \ker(\varGamma' \mp (-1)^{j+1})
=
\left\{
\begin{pmatrix}
\frac{a \mp (-1)^j}{b}\psi \\
\psi
\end{pmatrix} \mid
\psi \in \ker(L - \delta_{j,\pm})
\right\}, \qquad j=1,2.
\]
\begin{proof}
It follows from  \cref{equation: bounded assumption} that the following sequences are bounded:
\begin{equation}
\label{equation: bounded sequences}
\frac{-b^*}{a \mp 1} = \frac{a \pm 1}{b}, \qquad \frac{-q}{p \mp 1} = \frac{p \pm 1}{q^*}.
\end{equation}
For each  $\psi_1, \psi_2 \in \ell^2(\Z),$ let us first consider the following $\C^2$-vector;
\[
(\varGamma' \mp 1)
\begin{pmatrix}
\psi_1 \\
\psi_2
\end{pmatrix}
=
\begin{pmatrix}
a \mp 1 & b^* \\
b & -(a \pm 1)
\end{pmatrix}
\begin{pmatrix}
\psi_1 \\
\psi_2
\end{pmatrix} =
\begin{pmatrix}
(a \mp 1) \psi_1 +  b^*\psi_2 \\
b \psi_1 -(a \pm 1)\psi_2
\end{pmatrix}, 
\]
where $(a \mp 1) \psi_1 +  b^*\psi_2 = 0$ if and only if $b \psi_1 -(a \pm 1)\psi_2  = 0$ by the first equality in \cref{equation: bounded sequences}. It follows that the following equality holds true;
\[
\ker(\varGamma' \mp 1) =
\ker\left(
\begin{pmatrix}
a & b^* \\
b & -a
\end{pmatrix} 
\mp 1\right) =
\left\{
\begin{pmatrix}
 \frac{a \pm 1}{b} \psi \\
\psi
\end{pmatrix} \mid
\psi \in \ell^2(\Z)
\right\}.
\]
It follows that
\[
\ker(\varGamma \mp 1) =
\ker
\begin{pmatrix}
1 & 0\\
0  &  L^*
\end{pmatrix}
\left(
\begin{pmatrix}
p   & q \\
q^* & -p
\end{pmatrix}
\mp 1\right) 
\begin{pmatrix}
1 & 0\\
0  &  L
\end{pmatrix}
=
\left\{
\begin{pmatrix}
1 & 0\\
0  &  L^*
\end{pmatrix}
\begin{pmatrix}
\frac{-q}{p \mp 1}\psi \\
\psi
\end{pmatrix} \mid
\psi \in \ell^2(\Z)
\right\}, 
\]
where the last equality follows from the second equality in \cref{equation: bounded sequences}. We obtain the following equalities for each $j=1,2:$
\begin{align}
\label{equation: first part}
\ker(\varGamma' \mp (-1)^{j+1}) &=
\left\{
\begin{pmatrix}
\frac{a \mp (-1)^j}{b}\psi \\
\psi
\end{pmatrix} \mid
\psi \in \ell^2(\Z)
\right\},
% ------------------------------------ %
\\
\label{equation: second part}
\ker(\varGamma + (- 1)^{j}) &=
\left\{
\begin{pmatrix}
\frac{-q}{p + (- 1)^j}L\psi \\
\psi
\end{pmatrix} \mid
\psi \in \ell^2(\Z)
\right\}.
\end{align}
Next, we show that \cref{Equation: Linear Isomorphism} is a well-defined linear transform. Note that the bounded sequence $\delta_{j,\pm}$ consists of:
\begin{align*}
\Lambda((-1)^j p) &= \frac{1 + (-1)^j p}{1 - (-1)^j p} \times  \frac{1 + (-1)^j p}{1 + (-1)^j p}  = \frac{(1 + (-1)^j p)^2}{|q|^2}, \\
\Lambda(\mp (-1)^j a) &= \frac{1 \mp (-1)^j a}{1 \pm (-1)^j a} \times  \frac{1 \mp (-1)^j a}{1 \mp (-1)^j a}  = \frac{(1 \mp (-1)^j a)^2}{|b|^2},
\end{align*}
where $1 + (-1)^j p$ and $1 \mp (-1)^j a$ are sequences of positive numbers. We obtain
\begin{equation}
\label{equation: rearrangement of delta}
\delta_{j,\pm} 
= \pm e^{-i (\theta + \phi)}  \sqrt{\Lambda((-1)^j p)\Lambda(\mp (-1)^j a)} = \pm \frac{1 + (-1)^j p}{q} \frac{1 \mp (-1)^j a}{b} = \frac{p + (-1)^j}{-q}\frac{a \mp (-1)^j}{b}.
\end{equation}
It follows from \cref{equation: rearrangement of delta} that given $\psi \in \ell^2(\Z),$ we have that $\psi \in \ker(L - \delta_{j,\pm})$ if and only if the following equality holds true;
\[
\frac{-q}{p + (-1)^j}  L \psi = \frac{a \mp (-1)^j}{b} \psi.
\]
It follows from \crefrange{equation: first part}{equation: second part} that \cref{Equation: Linear Isomorphism} is a well-defined bijective linear transform.
\end{proof}

\begin{lemma}
\label{lemma1: first order difference equation}
Let $\delta = (\delta(x))_{x \in \Z}$ be a bounded sequence of non-zero complex numbers, and let
\[
\Delta := \sum_{x =1}^\infty \left(\prod_{y=1}^{x} |\delta(-y)|^{-2}\right) + \sum_{x =1}^\infty \left(\prod_{y=0}^{x-1} |\delta(y)|^2\right).
\]
Then the following assertions hold true:
\begin{enumerate}[(i)]
\item We have
\begin{equation}
\dim \ker\left(L -  \delta \right) =
\begin{cases}
1, & \Delta < \infty, \\
0, & \Delta = \infty.
\end{cases}
\end{equation}

\item  Let
\begin{align}
\delta^{\downarrow} &:= \min\left\{\liminf_{x \to \infty} \left(\prod_{y=1}^{x} |\delta(-y)|^{-2}\right)^{1/x}, \qquad \liminf_{x \to \infty} \left(\prod_{y=0}^{x-1} |\delta(y)|^2\right)^{1/x} \right\}, \\
\delta^{\uparrow} &:= \max\left\{\limsup_{x \to \infty} \left(\prod_{y=1}^{x} |\delta(-y)|^{-2}\right)^{1/x}, \qquad \limsup_{x \to \infty} \left(\prod_{y=0}^{x-1} |\delta(y)|^2\right)^{1/x} \right\}.
\end{align}
If $0 < \delta^{\downarrow} \leq \delta^{\uparrow} < 1,$ then $\dim \ker\left(L -  \delta \right) = 1.$ Moreover, in this case, for any $\epsilon > 0$ satisfying $0 < \delta^{\downarrow} - \epsilon  \leq \delta^{\uparrow} + \epsilon < 1,$ there exists $x_\epsilon \in \N,$ such that for any $\psi \in \ker\left(L -  \delta \right)$ we have
\begin{equation}
\label{equation1: exponential decay}
|\psi(0)|^2(\delta^{\downarrow} - \epsilon)^{|x|} \leq |\psi(x)|^2 \leq |\psi(0)|^2(\delta^{\uparrow} + \epsilon)^{|x|}, \qquad |x| \geq x_\epsilon.
\end{equation}
\end{enumerate}
\end{lemma}
Note that (i) shows that $\dim \ker\left(L -  \delta \right)$ depends only on $|\delta|.$ As for (ii), if $0 < \delta^{\downarrow} \leq \delta^{\uparrow} < 1,$ then  \cref{equation1: exponential decay} can be rewritten as
\[
|\psi(0)|^2e^{\log(\delta^{\downarrow} - \epsilon)|x|} \leq |\psi(x)|^2 \leq |\psi(0)|^2 e^{\log(\delta^{\uparrow} + \epsilon) |x|}, \qquad |x| \geq x_\epsilon,
\]
where $\epsilon > 0$ is any number satisfying $0 < \delta^{\downarrow} - \epsilon  \leq \delta^{\uparrow} + \epsilon < 1.$
\begin{proof}
\begin{comment}
We have
\begin{equation}
\label{equation: absolute value of delta}
|\delta|^2 
= \frac{(1 + \rho)^2}{1 - \rho^2}\frac{(1 + \alpha)^2}{1 - \alpha^2}
= \frac{1 + \rho}{1 - \rho}\frac{1 + \alpha}{1 - \alpha}
= \Lambda(\rho) \Lambda(\alpha).
\end{equation}
For each $x \in \Z$ we have 
\[
|\delta(x)|^2 = \Lambda(\rho(x)) \Lambda(\alpha(x)) \leq \Lambda\left(\sup_{x \in \Z} |\rho(x)|\right) \Lambda\left(\sup_{x \in \Z} |\alpha(x)|\right),
\]
where the last inequality follows from the fact that $\Lambda$ is an increasing sequence. It follows that $\delta$ is a bounded sequence.
\end{comment}

(i) We need to solve a difference equation of the form;
\begin{equation}
\label{equation: difference equation}
\psi(x+1) = \delta(x)\psi(x), \qquad \forall x \in \Z.
\end{equation}
Since each $\delta(x)$ is non-zero, such a solution is uniquely determined by the initial value $\psi(0).$ In particular, if $\psi, \psi' \in \ker\left(L -  \delta \right)$ are non-zero vectors, then $\psi(0),\psi'(0)$ are non-zero, and so the linear combination $\psi'(0)\psi - \psi(0)\psi'$ is the zero vector. It follows that $\psi,\psi'$ are linearly independent, and so $\dim \ker\left(L -  \delta \right) \leq 1.$ Suppose that we have a bounded sequence $\psi = (\psi(x))_{x \in \Z}$ satisfying \cref{equation: difference equation}. We have 
\begin{equation}
\label{equation: product formula for the solution}
\psi(x) = \prod_{y=0}^{x-1} \delta(y) \psi(0), \qquad 
\psi(-x) = \prod_{y=1}^{x} \delta(-y)^{-1} \psi(0), \qquad 
x \geq 1.
\end{equation}
Since $\sum_{x \in \Z} |\psi(x)|^2 =  |\psi(0)|^2 + \sum_{x \in \N} |\psi(-x)|^2  + \sum_{x \in \N} |\psi(x)|^2,$ we get
\[
\sum_{x \in \Z} |\psi(x)|^2  
= |\psi(0)|^2 + |\psi(0)|^2 \sum_{x \in \N} \prod_{y=1}^{x}|\delta(-y)|^{-2}    + |\psi(0)|^2 \sum_{x \in \N} \prod_{y=0}^{x-1} |\delta(y)|^2.
\]
That is, $\dim \ker(L - \delta) = 1$ if and only if $\Delta := \sum_{x =1}^\infty \left(\prod_{y=1}^{x} |\delta(-y)|^{-2}\right) + \sum_{x =1}^\infty \left(\prod_{y=0}^{x-1} |\delta(y)|^2\right) < \infty.$ 

(ii) If $0 < \delta^{\downarrow} \leq \delta^{\uparrow} < 1,$ then $\dim \ker\left(L -  \delta \right) = 1$ by the root test. Let $\epsilon > 0$ be any number satisfying $0 < \delta^{\downarrow} - \epsilon  \leq \delta^{\uparrow} + \epsilon < 1.$ It follows that there exists $x_\epsilon \in \N,$ such that 
\begin{align}
\delta^{\downarrow} - \epsilon &< \min\left\{\inf_{x \geq x_\epsilon} \left(\prod_{y=1}^{x} |\delta(-y)|^{-2}\right)^{1/x}, \inf_{x \geq x_\epsilon} \left(\prod_{y=0}^{x-1} |\delta(y)|^2\right)^{1/x} \right\}, \\
\delta^{\uparrow} + \epsilon &> \max\left\{\sup_{x \geq x_\epsilon} \left(\prod_{y=1}^{x} |\delta(-y)|^{-2}\right)^{1/x}, \sup_{x \geq x_\epsilon} \left(\prod_{y=0}^{x-1} |\delta(y)|^2\right)^{1/x} \right\}.
\end{align}
Let $\psi \in \ker(L - \delta),$ and let $|x| \geq x_\epsilon.$ On one hand, if $x \geq x_\epsilon,$ then $|\psi(x)|^2 = \prod_{y=0}^{x-1} |\delta(y)|^2 |\psi(0)|^2,$ and so
\[
(\delta^{\downarrow} - \epsilon)^x |\psi(0)|^2 < |\psi(x)|^2 < (\delta^{\uparrow} + \epsilon)^x |\psi(0)|^2.
\]
On the other hand, if $-x \geq x_\epsilon,$ then $|\psi(x)|^2 = \prod_{y=1}^{-x} |\delta(-y)|^{-2} |\psi(0)|^2,$ and so
\[
(\delta^{\downarrow} - \epsilon)^{-x} |\psi(0)|^2 < |\psi(x)|^2 < (\delta^{\uparrow} + \epsilon)^{-x} |\psi(0)|^2.
\]
The claim follows.
\end{proof}

\begin{proof}[Proof of \cref{theorem: bulk-boundary correspondence}]
It follows from \cref{theorem: suzuki formula}(i) that
\[
\ker(U \mp 1) = \bigoplus_{j=1,2} \ker(\varGamma + (-1)^j) \cap \ker(\varGamma' \mp (-1)^{j+1}).
\]
The linear isomorphisms of the form \cref{Equation: Linear Isomorphism} allow us to let
\[
m_{j,\pm} := \dim \left(\ker(\varGamma + (-1)^j) \cap \ker(\varGamma' \mp (-1)^{j+1}) \right) = \dim \ker(L - \delta_{j,\pm}), \qquad j=1,2.
\]
Let $\delta := \delta_{j,\pm},$ and let $\Delta := \Delta_{j,\pm}.$ It then follows from \cref{lemma1: first order difference equation} that 
\begin{equation}
m_{j,\pm}  =
\begin{cases}
1, & \Delta_{j, \pm} < \infty, \\
0, & \mbox{ otherwise.}
\end{cases}
\end{equation}
As in \cref{theorem: suzuki formula}(i), we obtain the following formulas:
\begin{align}
&\ind_\pm(\varGamma,U) = m_{1,\pm} - m_{2,\pm},\\
\label{equation: dimension of eigenspace} &\dim \ker(U \mp 1) = m_{1,\pm} + m_{2,\pm},
\end{align}
where each $m_{j,\pm}$ is either $0$ or $1.$ Assume the contrary that $\Delta_{j, \pm} < \infty$ for each $j=1,2.$ In this case, for each $j=1,2,$ we have $\prod_{y=0}^{x-1} |\delta_{j, \pm}(y)|^2 \to 0$ as $x \to \infty.$ Therefore, $\prod_{y=0}^{x-1} |\delta_{1, \pm}(y)|^2 |\delta_{2, \pm}(y)|^2 \to 0$ as $x \to \infty.$ Note, however, that this is impossible, since for each $y = 0, \dots, x-1$ we have
\[
|\delta_{1, \pm}(y)|^2 |\delta_{2, \pm}(y)|^2 = \Lambda(p(y))^{-1}\Lambda(\mp a(y))^{-1}\Lambda(p(y))\Lambda(\mp a(y)) = 1.
\]
This contradiction shows $\Delta_{1, \pm} + \Delta_{2, \pm} = \infty.$ 

(i) If $\Delta_{1, \pm} = \Delta_{2, \pm} = \infty,$ then we get the trivial equalities $\ind_\pm(\varGamma,U) = 0 = \dim \ker(U \mp 1).$ On one hand, if $\Delta_{1, \pm} < \infty,$ then $\ind_\pm(\varGamma,U) = 1 - 0$ and $\dim \ker(U \mp 1) =  1 + 0.$ On the other hand, if $\Delta_{2, \pm} < \infty,$ then $\ind_\pm(\varGamma,U) = 0 - 1$ and $\dim \ker(U \mp 1) = 0 + 1.$ Thus, the formulas \crefrange{equation: bulk-edge correspondence for ssqw}{equation: pm index for ssqw} have been verified.

(ii) It is obvious that \cref{Equation2: Linear Isomorphism} defines a linear isomorphism.

(iii) If $0 < \delta^{\downarrow}_{j,\pm} \leq \delta^{\uparrow}_{j,\pm} < 1$ for some $j=1,2,$ then $\Delta_{j, \pm} < \infty$ by the root test. Let $\epsilon > 0$ be any number satisfying $0 < \delta^{\downarrow}_{j,\pm} - \epsilon  < \delta^{\uparrow}_{j,\pm} + \epsilon < 1.$ It follows from \cref{lemma1: first order difference equation}(ii) that there exists $x_\pm \in \N,$ such that for any non-zero $\psi \in \ker(L - \delta_{j,\pm}),$ we have
\[
(\delta^{\downarrow}_{j, \pm} - \epsilon)^{|x|} \leq 
\frac{|\psi(x)|^2}{|\psi(0)|^2} \leq (\delta^{\uparrow}_{j, \pm} + \epsilon)^{|x|}, \qquad |x| \geq x_\epsilon.
\]
Let $\Psi := \tau_{j,\pm}(\psi)$ be defined by \cref{Equation2: Linear Isomorphism}. With \cref{Equation2: Linear Isomorphism} in mind, we have
$
\|\Psi(x)\|^2 
= (\Lambda(\mp (-1)^j a(x)) + 1 )|\psi(x)|^2.
$ for each $x \in \Z.$ The claim follows.
\end{proof}

% ------------------------------------------------------------------------------------------------------------ %
\subsection{The anisotropic case}

\begin{proof}[Proof of \cref{theorem: baby bulk-edge correspondence}]
Let $U = \Usuz$ be the evolution operator of Suzuki's split-step quantum walk given by \cref{equation: suzuki split-step quantum walk}, and let us assume the existence of the two-sided limits of the form \cref{equation: anisotropic assumption}. Let $|p(x)| < 1$ and $|a(x)| < 1$ for each $x \in \Z.$ It is shown in \cite[Theorem B(ii)]{Tanaka-2020} that 
\begin{align}
% ----------------------------------- %
\label{equation: essential spectrum of U}
\ess(U) &= \bigcup_{\star = \pm \infty} \left\{z \in \T \mid \Re z \in I(\star) \right\},  \\
% ----------------------------------- %
\label{equation: definition of Istar}
I(\star) &:= [p(\star) a(\star) - \sqrt{1 - p(\star)^2}\sqrt{1 - a(\star)^2}, p(\star) a(\star) + \sqrt{1 - p(\star)^2}\sqrt{1 - a(\star)^2}], \qquad \star = \pm \infty,
\end{align}
where $\T$ is the unit-circle in the complex plane. Since $p(\star), a(\star) \in (-1, 1),$ we can uniquely write $p(\star) = \sin \theta(\star)$ and $a(\star) = \sin \phi(\star)$ for some $\theta(\star), \phi(\star) \in (- \pi/2,\pi/2).$ We get
\[
p(\star) a(\star) \pm \sqrt{1 - p(\star)^2}\sqrt{1 - a(\star)^2} 
= \sin \theta(\star) \sin \phi(\star) \pm \cos \theta(\star) \cos \phi(\star)
= \pm \cos(\theta(\star) \mp \phi(\star)),
\]
where $- \pi < \theta(\star) \mp \phi(\star) < \pi.$ Thus $\pm 1 \in I(\star)$ if and only if $1 = \cos(\theta(\star) \mp \phi(\star))$ if and only if $\theta(\star) \mp \phi(\star) = 0.$ Since $(- \pi/2,\pi/2) \ni x \longmapsto \sin x \in (-1,1)$ is a bijective odd function, the last equality is equivalent to $p(\star) = \pm a(\star).$ That is, $p(\star) \neq  \pm a(\star)$ for each $\star  = \pm \infty$ if and only if $\pm 1 \notin \ess(U).$ From here on, we assume $\pm 1 \notin \ess(U),$ and prove \cref{theorem: baby bulk-edge correspondence}(i),(ii).

(i) Since $p(\star), a(\star) \in (-1,1),$ it follows from the continuity of $\Lambda$ that $\lim_{x \to \star} \Lambda(p(x)) \Lambda(\mp a(x)) =  \Lambda(p(\star)) \Lambda(\mp a(\star)).$ That is,
\[
\lim_{x \to \star} \Lambda(p(x)) \Lambda(\mp a(x)) \lessgtr 1 \mbox{ if and only if }  \iff p(\star) \mp a(\star) \lessgtr  0,
\]
where $\lessgtr$ simultaneously denotes the two binary relations $>$ and $<.$ For each $j=1,2,$ we get
\begin{align*}
\Delta_{j, \pm} 
%&= \sum_{x =1}^\infty \left(\prod_{y=1}^{x} |\delta_{j,\pm}(-y)|^{-2}\right) + \sum_{x =1}^\infty \left(\prod_{y=0}^{x-1} |\delta_{j,\pm}(y)|^2\right) \\
&= \sum_{x =1}^\infty \left(\prod_{y=0}^{x-1} (\Lambda(p(-y)) \Lambda(\mp a(-y)))^{(-1)^{j+1}}\right) + \sum_{x =1}^\infty \left(\prod_{y=0}^{x-1} (\Lambda(p(y)) \Lambda(\mp a(y)))^{(-1)^j}\right).
\end{align*}
It follows from the ratio test that 
\[
\Delta_{j, \pm} < \infty \mbox{ if and only if } (-1)^j(p(+\infty) \mp a(+\infty)) < 0 < (-1)^j(p(-\infty) \mp a(-\infty)), \qquad j=1,2.
\]
It is now easy to see that \cref{equation: pm index for ssqw} becomes \cref{equation: pm indices for anisotropic ssqw}.

(ii) We have
\begin{align}
&\delta^\downarrow_{j, \pm} = \min \left\{ (\Lambda(p(- \infty)) \Lambda(\mp a(- \infty)))^{(-1)^{j+1}}, (\Lambda(p(+\infty)) \Lambda(\mp a(+\infty)))^{(-1)^j}\right\}, \\
% ------ %
&\delta^\uparrow_{j, \pm} = \max \left\{ (\Lambda(p(- \infty)) \Lambda(\mp a(- \infty)))^{(-1)^{j+1}}, (\Lambda(p(+\infty)) \Lambda(\mp a(+\infty)))^{(-1)^j}\right\}.
\end{align}
It follows from \cref{theorem: bulk-boundary correspondence}(iii) that for each $j=1,2,$ we have $0 < \delta^{\downarrow}_{j,\pm} \leq \delta^{\uparrow}_{j,\pm} < 1$ if and only if $\Delta_{j, \pm} < \infty.$ Moreover, in this case, for any $\epsilon > 0$ satisfying $0 < \delta^{\downarrow}_{j,\pm} - \epsilon  < \delta^{\uparrow}_{j,\pm} + \epsilon < 1,$ there exists $x_\pm \in \N$ with the property that if $\psi \in \ker(L - \delta_{j,\pm})$ is a non-zero vector, then $\Psi := \tau_{j,\pm}(\psi)$ given by \cref{Equation2: Linear Isomorphism} exhibits the following exponential decay;
\begin{equation}
\Lambda_{j, \pm}^\downarrow \left(\delta^{\downarrow} _{j, \pm}- \epsilon \right)^{|x|} \leq \frac{\|\Psi(x)\|^2}{|\psi(0)|^2} \leq  \Lambda_{j, \pm}^\uparrow \left(\delta^{\uparrow}_{j, \pm} + \epsilon\right)^{|x|}, \qquad |x| \geq x_\pm.
\end{equation}
We obtain \cref{equation: exponential decay in the anisotropic case}, if we let 
\begin{align}
&\kappa^\downarrow_{j, \pm} := |\psi(0)|^2\Lambda^\downarrow_{j, \pm}, \\
% ------ %
&\kappa^\uparrow_{j, \pm} := |\psi(0)|^2\Lambda^\uparrow_{j, \pm}, \\
&c^\downarrow_{j, \pm} := -\log\left(\delta^{\downarrow}_{j, \pm} - \epsilon\right), \\
% ------ %
&c^\uparrow_{j, \pm} := -\log\left(\delta^{\uparrow}_{j, \pm} + \epsilon\right).
\end{align}

\end{proof}

% ------------------------------------------------------------------------------------------------------------ %
% New Section                                                                                                  %
% ------------------------------------------------------------------------------------------------------------ %
\section{Discussion and Concluding Remarks}
\label{section: concluding remarks}

% ------------------------------------------------------------------------------------------------------------ %
\subsection{Kitagawa's split-step quantum walk}

We show that \cref{equation: kitagawa split-step quantum walk} can be made unitarily equivalent to \cref{equation: suzuki split-step quantum walk}, provided that we appropriately define $p,q,a,b$ in terms of $\theta_1, \theta_2.$

\begin{lemma}
\label{lemma: equivalence of Ukit}
Let 
\[
p := \sin \theta_2(\cdot + 1), \qquad q := \cos \theta_2(\cdot + 1), \qquad 
a := -\sin \theta_1, \qquad b := \cos \theta_1.
\]
Then
\[
\begin{pmatrix}
0 & 1 \\
1 & 0 
\end{pmatrix} 
U_{\textnormal{kit}} 
\begin{pmatrix}
0 & 1 \\
1 & 0 
\end{pmatrix} 
= 
\begin{pmatrix}
p & q L \\
L^*q^*  & -p(\cdot - 1)
\end{pmatrix}
\begin{pmatrix}
a & b^* \\
b & -a
\end{pmatrix}.
\]
\end{lemma}
\begin{proof}
Let $\sigma_1$ be the first Pauli matrix. Given any $\R$-valued sequence $\theta = (\theta(x))_{x \in \Z},$ we consider the following rotation matrix;
\[
R(\theta) := 
\begin{pmatrix}
\cos \theta & - \sin \theta \\
\sin \theta & \cos \theta  
\end{pmatrix}.
\]
It is obvious that $R(\theta)R(\phi) = R(\theta + \phi)$ for any $\R$-valued sequences $\theta, \phi.$ If we let $\varGamma := \sigma_1 (1 \oplus L) R(\theta_2) (L^* \oplus 1)$ and $\varGamma' := R(\theta_1)\sigma_1,$ then
\[
\sigma_1 U_{\textnormal{kit}} \sigma_1 
= \varGamma \varGamma'.
\]
We have
\[
R(\theta_j) \sigma_1
=
\begin{pmatrix}
\cos \theta_j & - \sin \theta_j \\
\sin \theta_j & \cos \theta_j
\end{pmatrix}
\begin{pmatrix}
0 & 1 \\
1 & 0 
\end{pmatrix}
=
\begin{pmatrix}
 - \sin \theta_j &  \cos \theta_j\\
\cos \theta_j  &  \sin \theta_j
\end{pmatrix}.
\]
Now 
\[
\varGamma = \sigma_1 (1 \oplus L) R(\theta_2) (L^* \oplus 1) = 
\begin{pmatrix}
p & q L \\
L^*q^*  & -p(\cdot - 1)
\end{pmatrix}.
\]
\end{proof}

\subsection{Spectral gaps and decay rates}
We shall impose the assumption of \cref{theorem: baby bulk-edge correspondence} throughout this subsection. Recall that any non-zero vector$ \Psi := \tau_{j,\pm}(\psi)$ given by \cref{Equation2: Linear Isomorphism} exhibits the following exponential decay;
\begin{equation}
\Lambda_{j, \pm}^\downarrow \left(\delta^{\downarrow} _{j, \pm}- \epsilon \right)^{|x|} \leq \frac{\|\Psi(x)\|^2}{|\psi(0)|^2} \leq  \Lambda_{j, \pm}^\uparrow \left(\delta^{\uparrow}_{j, \pm} + \epsilon\right)^{|x|}, \qquad |x| \geq x_\pm,
\end{equation}
where 
\begin{align*}
&\delta^\downarrow_{j, \pm} = \min \left\{ (\Lambda(p(- \infty)) \Lambda(\mp a(- \infty)))^{(-1)^{j+1}}, (\Lambda(p(+\infty)) \Lambda(\mp a(+\infty)))^{(-1)^j}\right\}, \\
% ------ %
&\delta^\uparrow_{j, \pm} = \max \left\{ (\Lambda(p(- \infty)) \Lambda(\mp a(- \infty)))^{(-1)^{j+1}}, (\Lambda(p(+\infty)) \Lambda(\mp a(+\infty)))^{(-1)^j}\right\}.
\end{align*}

% ------------------------------------------------------------------------------------------------------------ %
\subsubsection{Half-gapped examples}

\begin{example}[half-gapped case]
Let $0 < p_0 < 1,$ and let
\begin{align*}
p(-\infty) &:= -p_0, & a(- \infty) &:= \pm p_0, \\
p(+\infty) &:= p_0, & a(+ \infty) &:= \mp p_0. 
\end{align*}
Since $\pm a(- \infty) = p_0 \neq p(- \infty)$ and $\pm a(+ \infty) = - p_0 \neq p(+ \infty),$ the essential spectrum of the operator $U$ has a spectral gap at $\pm 1.$ Moreover,
\[
p(-\infty) \mp a(-\infty)   = -2p_0 < 0 < 2p_0 =    p(+\infty) \mp a(+\infty).
\]
We get $\ind_\pm(\varGamma, U) = \dim \ker(U \mp 1) = 1.$ It follows from \crefrange{equation: essential spectrum of U}{equation: definition of Istar} that the essential spectrum of $U$ is given explicitly by
\begin{align*}
% ----------------------------------- %
\ess(U) =  \left\{z \in \T \mid \Re z \in I_\pm \right\},  \qquad 
% ----------------------------------- %
I_\pm := [\mp p_0^2 - (1 - p_0^2), \mp p_0^2 + (1 - p_0^2)].
\end{align*}
More precisely, the essential spectrum $\ess(U)$ can be classified into the following two distinct cases: 
\begin{figure}[H]
\centering

\begin{tikzpicture}
%--------- Spectrum 1 ----------- %
\begin{axis}[ticks=none, xmin=-2, xmax=2, ymin= -2, ymax=2, legend pos = north west, axis lines=center, xlabel=$\Re$, ylabel=$\Im$, xlabel style={anchor = north}
, width = 0.4\textwidth, height = 0.4\textwidth, clip=false
]
	\addplot [domain=0:2*pi,samples=50, smooth, densely dashed]({cos(deg(x))},{sin(deg(x))}); 
	\addplot [domain= 0: 1, samples=50, smooth, white]({cos(deg(6*pi/10))*x + (1-x)},{0}); 
	\addplot [domain= 0: 1, samples=50, smooth, densely dashed]({cos(deg(6*pi/10))*x + (1-x)},{0}); 
	\addplot [domain= pi: 6*pi/10, samples=50, smooth, line width=2.5pt]({cos(deg(x))},{sin(deg(x))}); 
	\addplot [domain= -pi: -6*pi/10, samples=50, smooth, line width=2.5pt]({cos(deg(x))},{sin(deg(x))}); 
    \node at (0,2.5) {(i) $(a(-\infty), a(+\infty)) = (p_0, -p_0)$};
    \node at (0,-2.5) {$\ess(U) =  \left\{z \in \T \mid \Re z \in I_+ \right\}$};
    \node at (1.2, -0.2) {$1$};
    \node at (0.2, 1.2) {$i$};
\end{axis}
\end{tikzpicture}
\qquad \qquad 
\begin{tikzpicture}
%--------- Spectrum 2 ----------- %
\begin{axis}[ticks=none, xmin=-2, xmax=2, ymin= -2, ymax=2, legend pos = north west, axis lines=center, xlabel=$\Re$, ylabel=$\Im$, xlabel style={anchor = north}
, width = 0.4\textwidth, height = 0.4\textwidth, clip=false
]
	\addplot [domain=0:2*pi,samples=50, smooth,  densely dashed]({cos(deg(x))},{sin(deg(x))}); 
	\addplot [domain= 0: 1, samples=50, smooth, white]({cos(deg(4*pi/10))*x - (1-x)},{0}); 
	\addplot [domain= 0: 1, samples=50, smooth,  densely dashed]({cos(deg(4*pi/10))*x - (1-x)},{0}); 
	\addplot [domain= -4*pi/10: 0, samples=50, smooth, line width=2.5pt]({cos(deg(x))},{sin(deg(x))}); 
	\addplot [domain= 0: 4*pi/10, samples=50, smooth, line width=2.5pt]({cos(deg(x))},{sin(deg(x))}); 
    \node at (0,2.5) {(ii) $(a(-\infty), a(+\infty)) = (-p_0, p_0)$};
    \node at (0,-2.5) {$\ess(U) =  \left\{z \in \T \mid \Re z \in I_- \right\}$};
    \node at (1.2, -0.2) {$1$};
    \node at (0.2, 1.2) {$i$};
\end{axis}
\end{tikzpicture}
\caption{
The black connected regions depict $\ess(U).$ We have $-1 \in \ess(U)$ in Case (i), whereas $+1 \in \ess(U)$ in Case (ii).
}
\label{figure: two spectra}
\end{figure}
We have $I_+ = [-1, 1 - 2p_0^2]$ in Case (i), and $I_- = [-1 + 2p_0^2, +1]$ in Case (ii). This motivates us to introduce the gap width $\Omega(p_0) := 2p_0^2.$  Moreover, we have
\[
\delta^\uparrow_{1, \pm} = \delta^\downarrow_{1, \pm} = \Lambda(- p_0)^2.
\] 
Note that $\Omega(p_0) = 2p_0^2$ increases as $p_0 \to 1.$ In this case, the convergence rates $\delta^\uparrow_{1, \pm} = \delta^\downarrow_{1, \pm}$ decrease, since $\Lambda(-p_0)^2 \to 0.$
\end{example}

% ------------------------------------------------------------------------------------------------------------ %
\subsubsection{Double-gapped examples}

We let $p(x) = 0$ for each $x \in \Z.$ We have the following index formula;
\begin{equation}
\label{equation: indices for the doubly gapped case}
\ind_\pm(\varGamma,U) = 
\begin{cases}
+1, & \mp a(-\infty) < 0 < \mp a(+\infty), \\
-1, & \mp a(+\infty) < 0 <  \mp a(-\infty), \\
0, & \mbox{otherwise}.
\end{cases}
\end{equation}
Note that we have
\begin{align}
\delta^{\downarrow}_{j, \pm} &= \min\left\{\Lambda(\mp (-1)^j a(+\infty))), \frac{1}{\Lambda(\mp (-1)^j a(-\infty)))}\right\}, \\
\delta^{\uparrow}_{j, \pm} &= \max\left\{\Lambda(\mp (-1)^j a(+\infty))), \frac{1}{\Lambda(\mp (-1)^j a(-\infty)))}\right\}, 
\end{align}
It follows from \crefrange{equation: essential spectrum of U}{equation: definition of Istar} that
\begin{align}
% ----------------------------------- %
\ess(U) &= \bigcup_{\star = \pm \infty} \left\{z \in \T \mid \Re z \in I(\star) \right\},  \\
% ----------------------------------- %
I(\star) &:= [- \sqrt{1 - a(\star)^2}, \sqrt{1 - a(\star)^2}], \qquad \star = \pm \infty.
\end{align}
\begin{figure}[H]
\centering
\begin{tikzpicture}
%--------- Spectrum 1 ----------- %
\begin{axis}[ticks=none, xmin=-2, xmax=2, ymin= -2, ymax=2, legend pos = north west, axis lines=center, xlabel=$\Re$, ylabel=$\Im$, xlabel style={anchor = north}
, width = 0.4\textwidth, height = 0.4\textwidth, clip=false
]
	\addplot [domain=0:2*pi,samples=50, smooth, densely dashed]({cos(deg(x))},{sin(deg(x))}); 
	\addplot [domain= 0: 1, samples=50, smooth, white]({cos(deg(3*pi/10))*x + (1-x)},{0}); 
	\addplot [domain= 0: 1, samples=50, smooth, densely dashed]({cos(deg(3*pi/10))*x + (1-x)},{0}); 
	\addplot [domain= 0: 1, samples=50, smooth, white]({cos(deg(7*pi/10))*x -(1-x)},{0}); 
	\addplot [domain= 0: 1, samples=50, smooth, densely dashed]({cos(deg(7*pi/10))*x -(1-x)},{0}); 	
	\addplot [domain= 3*pi/10: 7*pi/10, samples=50, smooth, line width=2.5pt]({cos(deg(x))},{sin(deg(x))}); 
	\addplot [domain= -3*pi/10: -7*pi/10, samples=50, smooth, line width=2.5pt]({cos(deg(x))},{sin(deg(x))}); 
%    \node at (0,2.5) {(i) $(a(-\infty), a(+\infty)) = (p_0, -p_0)$};
%    \node at (0,-2.5) {$\ess(U) =  \left\{z \in \T \mid \Re z \in I_+ \right\}$};
    \node at (1.2, -0.2) {$1$};
    \node at (0.2, 1.2) {$i$};
\end{axis}
\end{tikzpicture}
\caption{
This figure depicts $\ess(U).$ 
}
\label{figure: one spectrum}
\end{figure}
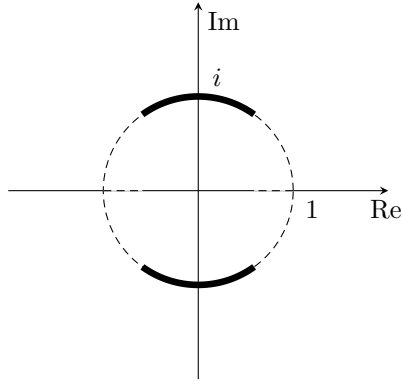
\cref{figure: one spectrum} motivates us to introduce the following gap width; 
\[
\Omega := \min \{1 - \sqrt{1 - a(-\infty)^2}, 1 - \sqrt{1 - a(+\infty)^2}\}.
\]
Suppose that $\mp a(-\infty) < 0 < \mp a(+\infty)$ holds true. Then $\Lambda(\mp a(-\infty)) < 1$ and $\frac{1}{\Lambda(\mp a(+\infty))} < 1.$ Note that $\Omega$ increases as $|a(\star)| \to 1$ for each $\star = \pm \infty.$ This corresponds to $\mp a(-\infty) \to -1$ and $\mp a(+\infty) \to +1.$ Thus, $\Lambda(\mp a(-\infty)) \to 0$ and $\frac{1}{\Lambda(\mp a(+\infty))} \to 0.$ It follows that $\delta^{\uparrow}_{2, \pm} \to 0$ and $\delta^{\downarrow}_{2, \pm} \to 0.$

% ------------------------------------------------------------------------------------------------------------ %
\subsection{The square of the evolution operator}

\cref{theorem: bulk-boundary correspondence} gives a concrete quantum walk example with the property that the estimate \cref{equation: topological protection of bounded states} becomes an equality. The purpose of the current subsection is to show that this is not always the case. Our counter example is based on the following simple proposition. 

\begin{proposition}
If $(\varGamma, U)$ is an abstract chiral pair on a Hilbert space $\cH,$ then $(\varGamma, U^2)$ and $(\varGamma'\varGamma \varGamma', U^2)$ are unitarily equivalent chiral pairs. Moreover, the following assertions hold true:
\begin{enumerate}[(i)]
\item If $\ker(U^2 - 1) = \ker(U - 1) \oplus \ker(U + 1)$ is finite-dimensional, then
\begin{align}
%\label{Equation: Want to show}
\ind_+(\varGamma,U^2) &= \ind(\varGamma, U).
\end{align}

\item If $\ker(U^2 + 1) = \ker(U - i) \oplus \ker(U + i)$ is finite-dimensional, then
\begin{align}
%\label{Equation: Want to show}
\ind_-(\varGamma,U^2) &= 0.
\end{align}

\item If $\ker(U^2 - 1) \oplus \ker(U^2 + 1)$ is finite-dimensional, then $\ind(\varGamma, U^2) = \ind(\varGamma,U).$
\end{enumerate}
\end{proposition}
\begin{proof}
Note that $(\varGamma, U^2)$ and $(\varGamma'\varGamma \varGamma', U^2)$ are chiral pairs, since $U^2 = \varGamma (\varGamma'\varGamma \varGamma').$ We have
\begin{equation}
\label{equation: unitary invariance for the square of U}
(\varGamma, U^2) 
= (\varGamma, \varGamma\varGamma'\varGamma\varGamma') 
\cong (\varGamma' \varGamma \varGamma', \varGamma'(\varGamma\varGamma'\varGamma\varGamma')\varGamma')
= (\varGamma' \varGamma \varGamma', \varGamma' \varGamma\varGamma'\varGamma)
= (\varGamma' \varGamma \varGamma',(U^2)^*) 
\cong (\varGamma' \varGamma \varGamma', U^2),
\end{equation}
where $\cong$ represents unitary equivalence. If $U$ admits the standard representation of the form \cref{equation: standard representation of U}, then $U^2$ admits the following standard representation;
\begin{equation}
U^2 =
\begin{pmatrix}
2R_1^2-1 & 2i Q_2R_2 \\
2i Q_1R_1  & 2R_2^2-1 \\
\end{pmatrix}.
\end{equation}
It follows that
\begin{equation}
\label{equation: indices for the square of U}
\ind_\pm(\varGamma, U^2) := \dim \ker ((2R_1^2 - 1) \mp 1) - \dim \ker ((2R_2^2-1) \mp 1).
\end{equation}
(i) If $\ker(U^2 - 1) = \ker(U - 1) \oplus \ker(U + 1)$ is finite-dimensional, then $\ker(U - 1) = \ker(R - 1)$ and $\ker(U + 1) = \ker(R + 1)$ are finite-dimensional. In this case,
\begin{align*}
\ind_+(\varGamma, U^2)
&= \dim \ker((2R_1^2 - 1) - 1) - \dim \ker((2R_2^2 - 1) - 1) \\
&= \dim \ker(R_1^2 - 1) - \dim \ker(R_2^2 - 1) \\
&= \dim \ker(R_1 - 1) + \dim \ker(R_1 + 1) - (\dim \ker(R_2 - 1) + \dim \ker(R_2 - 1))\\
&= \dim \ker(R_1 - 1) - \dim \ker(R_2 - 1) + \dim \ker(R_1 + 1) - \dim \ker(R_2 - 1)\\
&= \ind_+(\varGamma,U) + \ind_-(\varGamma,U) \\
&= \ind(\varGamma,U).
\end{align*}

(ii) If $\ker(U^2 + 1) = \ker(U - i) \oplus \ker(U + i)$ is finite-dimensional, then it follows from \cref{equation: unitary invariance for the square of U} that
\[
\ind_-(\varGamma, U^2) = \ind_-(\varGamma' \varGamma \varGamma', U^2) = -\ind_-(\varGamma, U^2),
\]
where the last equality follows from \cref{equation: index formulas with different chiral symmetry conditions}. We get $\ind_-(\varGamma, U^2) = 0.$

(iii) This follows from (i) and (ii).
\end{proof}

\begin{remark}
If $\ker(U^2 + 1) = \ker(U - i) \oplus \ker(U + i)$ is finite-dimensional, then $\ind_-(\varGamma, U^2) = 0.$ It follows from \cref{equation: indices for the square of U} that $\ker R_1$ and $\ker R_2$ have the same finite dimension, say, $n.$ We get $\dim(\ker(U^2 + 1)) = 2 n.$
\end{remark}

\begin{example}
\label{proposition: counter example}
Let $U$ be the evolution operator of Suzuki's split-step quantum walk, and let $p(x) = 0$ for each $x \in \Z.$ Suppose that there exists $a_0 \in (0,1)$ with the property that $a(x) \to \pm a_0$ as $x \to \pm \infty.$ We make use of the index formula \cref{equation: indices for the doubly gapped case}. Since $a(-\infty) = -a_0 < 0 < a_0 = a(+\infty),$ we have $\ind_+(\varGamma, U) = -1$ and $\ind_-(\varGamma, U) = 1.$ Thus $\ind_+(\varGamma, U^2) = 0.$ On the other hand, we get
\[
\dim \ker(U^2 - 1) 
= \dim(\ker(U - 1) \oplus \ker(U + 1)) 
= |\ind_+(\varGamma, U)| + |\ind_-(\varGamma, U)|
=2.
\]
Then $\dim \ker(U^2 - 1) = 2$ and $\ind_+(\varGamma, U^2) = 0.$ That is, $|\ind_+(\varGamma, U^2)| \neq \dim \ker(U^2 - 1).$
\end{example}

% ------------------------------------------------------------------------------------------------------------ %
\subsection{A new derivation of the existing index formulas}

The purpose of the current subsection is to give an alternative derivation of the following existing index formulas by making use of \cref{equation: pm indices for anisotropic ssqw};
\begin{theorem}[{\cite[Theorem B]{Tanaka-2020}}]
Under the assumption of \cref{theorem: baby bulk-edge correspondence}, let $-1, +1 \notin \ess(U).$ Then
\begin{align}
% --------------------------- %
\label{equation1: witten index for ssqw}
\ind (\varGamma, U) &=
\begin{cases}
0, &  |p(-\infty)| < |a(-\infty)| \mbox{ and }  |p(+\infty)| < |a(+\infty)|, \\
+\sign p(+\infty), &  |p(-\infty)| < |a(-\infty)| \mbox{ and }  |p(+\infty)| > |a(+\infty)|, \\
-  \sign p(-\infty), &  |p(-\infty)| > |a(-\infty)| \mbox{ and } |p(+\infty)| < |a(+\infty)|, \\
+\sign p(+\infty) - \sign p(-\infty), &  |p(-\infty)| >  |a(-\infty)| \mbox{ and }  |p(+\infty)| > |a(+\infty)|,
\end{cases} \\
% --------------------------- %
\label{equation2: witten index for ssqw}
\ind (\varGamma', U) &=
\begin{cases}
-\sign a(+\infty) + \sign a(-\infty), &  |p(-\infty)| < |a(-\infty)| \mbox{ and }  |p(+\infty)| < |a(+\infty)|, \\
+ \sign a(-\infty), &  |p(-\infty)| < |a(-\infty)| \mbox{ and }  |p(+\infty)| > |a(+\infty)|, \\
- \sign a(+\infty), &  |p(-\infty)| > |a(-\infty)| \mbox{ and }  |p(+\infty)| < |a(+\infty)|, \\
0, &  |p(-\infty)| >  |a(-\infty)| \mbox{ and }  |p(+\infty)| > |a(+\infty)|,
\end{cases}
% --------------------------- %
\end{align}
where $\sign$ denotes the sign function $\sign : \R \to \{-1,1\}.$ We let $\sign 0 := 1$ by convention.
\end{theorem}
\begin{proof}
Let $i_+$ (resp. $i_-$) be defined by the right hand side of \cref{equation1: witten index for ssqw} (resp. of \cref{equation2: witten index for ssqw}). We are required to show $i_+ = \ind(\varGamma,U)$ and $i_- = \ind(\varGamma',U).$ Note first that \cref{equation: pm indices for anisotropic ssqw} can be rewritten as 
\[
\ind_\pm(\varGamma,U) = 
\begin{cases}
0,  & p(-\infty) <  \pm a(-\infty) \mbox{ and } p(+\infty) <  \pm a(+\infty), \\
+1, & p(-\infty) <  \pm a(-\infty) \mbox{ and } p(+\infty) >  \pm a(+\infty), \\
-1, & p(-\infty) >  \pm a(-\infty) \mbox{ and } p(+\infty) < \pm a(+\infty), \\
0, & p(-\infty) >  \pm a(-\infty) \mbox{ and } p(+\infty) >  \pm a(+\infty).
\end{cases}
\]
Note also that we have $i_\pm = \ind_+(\varGamma,U) \pm \ind_-(\varGamma,U)$ in each of the $16$ cases defined by \cref{table: classification}. The claim follows from \cref{theorem: suzuki formula}(iii).
\end{proof}

{
\footnotesize
\begin{table}[H]
\begin{center}
\begin{tabular}{|c|c|c|c|c|c|c|c| } 
\hline
Cases &\multicolumn{1}{|c|}{$-\infty$} & \multicolumn{1}{|c|}{$+\infty$} & $\ind_+(\varGamma,U)$ & $\ind_-(\varGamma,U)$ & $i_+$ & $i_-$ \\ \hline
Case 1& $|p(-\infty)| < +a(-\infty)$  & $|p(+\infty)| < +a(+\infty)$   & $ 0$  & $ 0$  & $ 0$ & $ 0$  \\ %\hline
Case 2& $|p(-\infty)| < +a(-\infty)$  & $|p(+\infty)| < -a(+\infty)$  & $+1$  & $-1$  & $ 0$ & $+2$ \\ %\hline
Case 3&$|p(-\infty)| < +a(-\infty)$  & $+p(+\infty) > |a(+\infty)|$   & $+1$  & $ 0$  & $+1$ & $+1$  \\ %\hline
Case 4&$|p(-\infty)| < +a(-\infty)$  & $-p(+\infty) > |a(+\infty)|$  & $ 0$  & $-1$  & $-1$ & $+1$  \\ %\hline
Case 5&$|p(-\infty)| < -a(-\infty)$ & $|p(+\infty)| < +a(+\infty)$   & $ -1$ & $+1$  & $ 0$ & $-2$  \\ %\hline
Case 6&$|p(-\infty)| < -a(-\infty)$ & $|p(+\infty)| < -a(+\infty)$  & $ 0$  & $ 0$  & $ 0$ & $ 0$  \\ %\hline
Case 7&$|p(-\infty)| < -a(-\infty)$ & $+p(+\infty) > |a(+\infty)|$   & $ 0$  & $+1$  & $+1$ & $-1$  \\ %\hline
Case 8&$|p(-\infty)| < -a(-\infty)$ & $-p(+\infty) > |a(+\infty)|$  & $-1$  & $ 0$  & $-1$ & $-1$ \\ %\hline
Case 9&$+p(-\infty) > |a(-\infty)|$  & $|p(+\infty)| < +a(+\infty)$   & $-1$  & $ 0$  & $-1$ & $-1$  \\% \hline
Case 10&$+p(-\infty) > |a(-\infty)|$  & $|p(+\infty)| < -a(+\infty)$  & $ 0$  & $-1$  & $-1$ & $+1$ \\ %\hline
Case 11&$+p(-\infty) > |a(-\infty)|$  & $+p(+\infty) > |a(+\infty)|$   & $ 0$  & $ 0$  & $0 $ & $ 0$  \\ %\hline
Case 12&$+p(-\infty) > |a(-\infty)|$  & $-p(+\infty) > |a(+\infty)|$  & $-1$  & $-1$  & $-2$ & $ 0$  \\ %\hline
Case 13&$-p(-\infty) > |a(-\infty)|$ & $|p(+\infty)| < +a(+\infty)$   & $ 0$  & $+1$  & $+1$ & $-1$  \\ %\hline
Case 14&$-p(-\infty) > |a(-\infty)|$ & $|p(+\infty)| < -a(+\infty)$  & $+1$  & $ 0$  & $+1$ & $+1$  \\ %\hline
Case 15&$-p(-\infty) > |a(-\infty)|$ & $+p(+\infty) > |a(+\infty)|$   & $+1$  & $+1$  & $+2$ & $ 0$ \\ %\hline
Case 16&$-p(-\infty) > |a(-\infty)|$ & $-p(+\infty) > |a(+\infty)|$  & $ 0$  & $ 0$  & $ 0$ & $ 0$  \\ \hline
\end{tabular}
\caption{Classification of the indices}
\label{table: classification}
\end{center}
\end{table}
}

% ------------------------------------------------------------------------------------------------------------ %
\section*{Acknowledgements}

The authors would like to thank the members of the Shinshu Mathematical Physics Group for extremely useful comments and discussions. Our sincere thanks also go to D.~Funakawa, K.~Saito, and K.~Wada for giving us their valuable feedback. Y.T. acknowledges support by JSPS KAKENHI Grant Number 20J22684. This work was also partially supported by the Research Institute for Mathematical Sciences, an International Joint Usage/Research Center located in Kyoto University.

\bibliographystyle{alpha}
\bibliography{Bibliography}

\end{document}